\documentclass[preprint,12pt,3p]{elsarticle}

\usepackage{amssymb}

\usepackage{amsthm}
\usepackage{graphicx}
\usepackage{amsmath}
\usepackage{hyperref}

\usepackage{pgf,tikz}
\usepackage{float}

\usetikzlibrary{arrows}
\usetikzlibrary{shapes}
\usepackage{multirow}
\usepackage{latexsym}
\usepackage{setspace}
\PassOptionsToPackage{boxed,section}{algorithm}
\usepackage{algorithm}
\usepackage{algorithmic}
\usepackage{enumerate}
\usepackage{amsmath}			
\usepackage{amssymb}			
\usepackage{url}
\usepackage{setspace}
\usepackage{tikz}
\usetikzlibrary{arrows}
\usetikzlibrary{shapes}
\usetikzlibrary{decorations.markings}
\usepackage{geometry}
\usepackage{bbm}

\usepackage{paralist}

\newtheorem{proposition}{\em Proposition}
\newtheorem{theorem}{\em Theorem}
\newtheorem{conjecture}{\em Conjecture}
\newtheorem{definition}{\em Definition}

\journal{Sample Journal}

\begin{document}

\begin{frontmatter}

\title{Normal $5$-edge-colorings of a family of Loupekhine snarks}

\author[label1]{Luca Ferrarini}
\address[label1]{Dipartimento di Informatica,
Universita degli Studi di Verona, Strada le Grazie 15, 37134 Verona, Italy}

\cortext[cor1]{Corresponding author}

\ead{luca.ferrarini\_01@studenti.univr.it}

\author[label1]{Giuseppe Mazzuoccolo\corref{cor1}}

\ead{giuseppe.mazzuoccolo@univr.it}

\author[label1]{Vahan Mkrtchyan}
\ead{vahanmkrtchyan2002@ysu.am}


\begin{abstract}
In a proper edge-coloring of a cubic graph an edge $uv$ is called poor or rich, if the set of colors of the edges incident to $u$ and $v$ contains exactly three or five colors, respectively. An edge-coloring of a graph is normal, if any edge of the graph is either poor or rich. 
In this note, we show that some snarks constructed by using a method introduced by Loupekhine admit a normal edge-coloring with five colors. The existence of a Berge-Fulkerson Covering for a part of the snarks considered in this paper was recently proved by Manuel and Shanthi (2015). Since the existence of a normal edge-coloring with five colors implies the existence of a Berge-Fulkerson Covering, our main theorem can be viewed as a generalization of their result. 

\end{abstract}


\end{frontmatter}



\section{Introduction}
\label{sec:intro}

In graph theory the Petersen coloring conjecture asserts that the edges of every bridgeless cubic graph $G$ can be colored by using as set of colors the edge-set of the Petersen graph $P$, such that adjacent edges of $G$ are colored with adjacent edges of $P$. The conjecture is well-known and is considered hard to prove as it implies classical conjectures in the field such as the Berge-Fulkerson conjecture and the Cycle Double Cover conjecture (see \cite{Fulkerson,Jaeger1985,Zhang1997}). Jaeger in \cite{Jaeger1985} introduced an equivalent formulation of the Petersen coloring conjecture. More precisely, he showed that a bridgeless cubic graph satisfies this conjecture, if and only if, it admits a normal edge-coloring with at most five colors (see exact definitions later). Let $\chi'_{N}(G)$ denote the minimum number of colors in a normal edge-coloring of $G$. Usually, it is called the normal chromatic index of $G$. In this terms, the Petersen coloring conjecture amounts to proving that every bridgeless cubic graph has normal chromatic index at most five. To the best of our knowledge, the smallest known upper bound for $\chi'_{N}(G)$ in an arbitrary bridgeless cubic graph is $7$ (see \cite{Bilkova12,Normal7flows}). 
The situation is similar in the larger class of all simple cubic graphs (not necessarily bridgeless): there are cubic graphs with normal chromatic index $7$, on the other hand, in \cite{Normal7flows} two of the authors have shown that any simple cubic graph admits a normal $7$-edge-coloring. 
One may wonder whether the upper bound $7$ can be improved in other interesting subclasses of cubic graphs. Related with this question, in this paper we show that a class of snarks obtained with a method introduced by Loupekhine, admits a normal edge-coloring with five colors. 
Let us remark that this result implies the main result in \cite{ManSha}.

Now, let us give the all necessary definitions and notions used in the paper. We consider finite and undirected graphs. They do not contain neither loops nor parallel edges. 

For a graph $G$, the set of vertices and edges of $G$ are denoted by $V(G)$ and $E(G)$, respectively. Let $\partial_{G}(v)$ be the set of edges of $G$ that are incident to the vertex $v$ of $G$. Assume that $G$ and $H$ are two cubic graphs. If there is a mapping $\phi:E(G)\rightarrow E(H)$, such that for each $v\in V(G)$ there is $w\in V(H)$ such that $\phi(\partial_{G}(v)) = \partial_{H}(w)$, then we refer to $\phi$ as an $H$-coloring of $G$. If $G$ admits an $H$-coloring, then we write $H\prec G$. 



    \begin{figure}[ht]
	\begin{center}
	\begin{tikzpicture}[style=thick]
\draw (18:2cm) -- (90:2cm) -- (162:2cm) -- (234:2cm) --
(306:2cm) -- cycle;
\draw (18:1cm) -- (162:1cm) -- (306:1cm) -- (90:1cm) --
(234:1cm) -- cycle;
\foreach \x in {18,90,162,234,306}{
\draw (\x:1cm) -- (\x:2cm);
\draw[fill=black] (\x:2cm) circle (2pt);
\draw[fill=black] (\x:1cm) circle (2pt);
}
\end{tikzpicture}
	\end{center}
	\caption{The graph $P_{10}$.}\label{fig:Petersen10}
\end{figure}
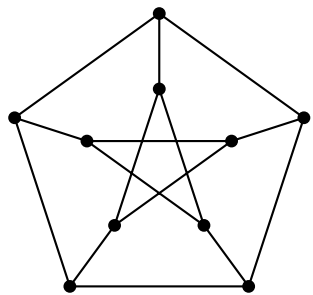

Let $P_{10}$ be the Petersen graph (Figure \ref{fig:Petersen10}). The Petersen coloring conjecture of Jaeger states:
\begin{conjecture}\label{conj:P10conj} (Jaeger, 1988 \cite{Jaeger1988}) For any bridgeless cubic
graph $G$, we have $P_{10} \prec G$.
\end{conjecture}

It is shown in \cite{Mkrt2013} that the Petersen graph is the only connected, bridgeless cubic graph that can color all bridgeless cubic graphs. The conjecture is difficult to prove, since it implies the classical Berge-Fulkerson conjecture \cite{Fulkerson,Seymour} and (5,2)-cycle-cover conjecture \cite{Celmins1984,Preiss1981}.
For our aim, we only need to recall the statement of the former one.

\begin{conjecture}\label{conj:BF} (Berge-Fulkerson, 1972 ) Any bridgeless cubic graph $G$ contains six (not necessarily distinct) perfect matchings such that any edge of $G$ belongs to exactly two of them.
\end{conjecture}

We will call Berge-Fulkerson Covering of $G$ any set of six perfect matchings which satisfy the condition in Conjecture \ref{conj:BF}. 

\section{Normal 5-edge-colorings}

A $k$-edge-coloring of a graph $G$ is an assignment of colors $\{1,...,k\}$ to edges of $G$, such that adjacent edges receive different colors. For an edge-coloring $c$ of $G$ and a vertex $v$ of $G$, let $S_{c}(v)$ be the set of colors that edges incident to $v$ receive. 

\begin{definition}\label{def:poorrich}
Assume that $uv$ is an edge of a cubic graph $G$, and let $c$ be an edge-coloring of $G$. We say that the edge $uv$ is {\bf poor} or {\bf rich} with respect to $c$, if $|S_{c}(u)\cup S_{c}(v)|=3$ or $|S_{c}(u)\cup S_{c}(v)|=5$, respectively. An edge is {\bf normal} (with respect to $c$) if it is poor or rich.
\end{definition}

Edge-colorings in which all edges are poor are trivially $3$-edge-colorings of $G$. On the other hand, edge-colorings having only rich edges have been considered before, and they are called strong edge-colorings \cite{Andersen1992}. In the present paper, we focus on the case when all edges must be normal.

\begin{definition}\label{def:normal}
An edge-coloring $c$ of a cubic graph is {\bf normal}, if any edge is normal with respect to $c$. 
\end{definition} 

In \cite{Jaeger1985}, Jaeger has shown that:

\begin{proposition}\label{prop:JaegerNormalColor}(Jaeger, \cite{Jaeger1985}) Let $G$ be any cubic graph. Then $P_{10}\prec G$, if and only if $G$ admits a normal $5$-edge-coloring.
\end{proposition} 

This implies that Conjecture \ref{conj:P10conj} can be re-stated in the following way:

\begin{conjecture}\label{conj:5NormalConj} For any bridgeless cubic graph $G$, $\chi'_{N}(G)\leq 5$.
\end{conjecture} 

Our previous considerations imply that Conjecture \ref{conj:5NormalConj} holds for $3$-edge-colorable cubic graphs. This means that the non-$3$-edge-colorable cubic graphs are the main obstacle for proving Conjecture \ref{conj:5NormalConj}. Let us note that in \cite{HaggSteff2013} Conjecture \ref{conj:5NormalConj} is verified for some non-$3$-edge-colorable bridgeless cubic graphs. Also, in \cite{Samal2011} the percentage of edges of a bridgeless cubic graph, which can be made normal in a 5-edge-coloring, is estimated.

The following well-known result says that Conjecture \ref{conj:5NormalConj} implies Conjecture \ref{conj:BF}.

\begin{proposition}\label{PimpliesBF}
Let $G$ be a bridgeless cubic graph. If $G$ admits a normal $5$-edge-coloring, then it admits a Berge-Fulkerson covering.  
\end{proposition}
\begin{proof}
{ \it (Sketch)} The existence of a normal $5$-edge-coloring is equivalent to the existence of a Petersen Coloring $\phi:E(G)\rightarrow E(P)$ of $G$. The preimages of the six distinct perfect matchings of $P$ give a Berge-Fulkerson covering of $G$ (see also \cite{Mkrt2013}). 
\end{proof}

\section{Loupekhine snarks}
\label{sec:aux}

In many contexts, among all non-$3$-edge-colorable cubic graphs, ones with some additional restrictions are more studied in order to avoid trivial cases. Even if there are many interesting definitions of what non-trivial could mean (see for instance \cite{FiMaSt,NedSko}), here we stick to the largely used definition of snark as a non-$3$-edge-colorable cyclically $4$-edge-connected cubic graph of girth at least $5$. 

A classical method to construct snarks was proposed by Loupekhine in early 70s and it appears firstly in a paper of Isaacs in 1976 (see \cite{Isaacs1976}).

Here, we recall the general method and we focus our attention on some families of snarks which can be constructed by using such a method. In particular, we will make use of the same terminology and notations in \cite{Loupekhine}. 

Let $G$ be a snark. We denote by $B(G)$ any induced subgraph of $G$ obtained by removing the three vertices of a path of length $2$ in $G$. Since $G$ has girth at least $5$, the subgraph $B(G)$ has exactly $5$ vertices of degree $2$, while every other vertex has degree equal to $3$.

Let $G_1,G_2,\ldots,G_k$ be arbitrary snarks. Consider the disjoint union of the $k$ graphs $B_i=B(G_i)$. For all $i=1,\ldots,k$, add two new edges between two of the vertices of degree $2$ in $B_i$ with two of the vertices of degree $2$ in $B_{i+1}$ (indices larger than $k$ are taken modulo $k$) in such a way that the set of these new edges is a matching in the resulting graph $G^B_k$.    

The new graph $G_k^B$ has exactly $k$ vertices of degree $2$, say $w_1, w_2,\ldots,w_k$.
Let $G_k^C$ be a graph with exactly $k$ vertices of degree one, say $z_1,z_2,\ldots,z_k$, and all other vertices of degree $3$.

If we identify every vertex $w_i$ of $G_k^B$ with a vertex $z_i$ of $G_k^C$, we obtain a cubic graph, denoted by $G^L_k$. 

\begin{proposition}(\cite{Isaacs1976})
If $k$ odd, then $G^L_k$ is a snark.
\end{proposition}

All snarks obtained with the method described in previous proposition are known as Loupekhine snarks, or $L$-snarks.

\begin{definition}
Let $G$ be an $L$-snark. 
\begin{itemize}
 \item  If every connected component of $G_k^C$ is either an isolated edge or a star with three vertices of degree $1$, then $G$ is said to be an $L_1$-snark.
 \item If every subgraph $B_i$ is isomorphic to $B(P)$ (see Figure \ref{fig:unified}), where $P$ is the Petersen graph, then $G$ is said to be an $LP$-snark.  
 \item An $LP_1$-snark is a snark which is both an $L_1$-snark and an $LP$-snark.
\end{itemize}
  
\end{definition}

In \cite{ManSha}, the authors prove that a very special family of $LP_1$-snarks admits a Berge-Fulkerson covering. In the next section, we strongly generalize their result by proving the existence of a normal $5$-edge-coloring for a larger class of $LP_1$-snarks which properly contains the class of $LP_1$-snarks studied in \cite{ManSha}. Hence, by Proposition \ref{PimpliesBF}, we have that all of them admit also a Berge-Fulkerson covering. 

	\begin{figure}[!htbp]
		\centering
		\begin{tikzpicture}[scale=0.60]
		\coordinate (j) at (0,0);
		\coordinate (k)  at (5,0) ;
		\coordinate (v) at (1,1.5) ;
		\coordinate (w) at (4,1.5) ;
		\coordinate (z) at (4.5,4) ;
		\coordinate (x) at (0.5,4) ;
		\coordinate (x1) at (-1.5,5.2);
		\coordinate (z1) at (6.5,5.2);
		\draw[fill=black] (x1) circle [radius=0.15cm];
		\draw[fill=black] (z1) circle [radius=0.15cm];
		\coordinate  (y) at (2.5,5.5);
		\draw[fill=black] (x) circle [radius=0.15cm];
		\draw[fill=black] (y) circle [radius=0.15cm];
		\draw[fill=black] (z) circle [radius=0.15cm];
		\draw[fill=black] (w) circle [radius=0.15cm];
		\draw[fill=black] (v) circle [radius=0.15cm];
		\draw[fill=black] (j) circle [radius=0.15cm];
		\draw[fill=black] (k) circle [radius=0.15cm];
		\coordinate (j') at (9+3,0);
		\coordinate (k') at (14+3,0);
		\coordinate (v') at (10+3,1.5);
		\coordinate (w') at (13+3,1.5);
		\coordinate (z') at (13.5+3,4);
		\coordinate (x') at (9.5+3,4);
		\coordinate (y') at (11.5+3,5.5);
		\coordinate (n') at (2.5, 7.5);
		\draw[fill=black] (j') circle [radius=0.15cm];
		\draw[fill=black] (x') circle [radius=0.15cm];
		\draw[fill=black] (y') circle [radius=0.15cm];
		\draw[fill=black] (z') circle [radius=0.15cm];
		\draw[fill=black] (w') circle [radius=0.15cm];
		\draw[fill=black] (v') circle [radius=0.15cm];
		\draw[fill=black] (j') circle [radius=0.15cm];
		\draw[fill=black] (k') circle [radius=0.15cm];
		\draw[fill=black] (n') circle [radius=0.15cm];
		\draw (v)--(j)--(k)--(w)--(y)--(v)--(z)--(x)--(w);
		\draw (v')--(j')--(k')--(w')--(y')--(v')--(z')--(x')--(w');
		\draw[dashed] (n')--(x1);
		\draw[dashed] (n')--(z1)--(k);
		\draw[dashed] (n')--(y);
		\draw[dashed] (x1)--(j);
		\draw[dashed] (x)--(x1);
		\draw[dashed] (z)--(z1);
		
		\draw (z')--(13.5+5,4);
		\draw (k')--(16+3,0);
		\draw (x')--(10.5,4);
		\draw (j')--(10,0);
		\draw (y')--(11.5+3,7.6);
		\node[anchor=north west] at (2.1,8.4) {$\mathbf{b}$};
		\node[anchor=north west] at (-1.7,6) {$\mathbf{a}$};
		\node[anchor=north west] at (6.5,5.8) {$\mathbf{c}$};
		;
		
		\end{tikzpicture}
		
	\caption{The component $B(P)$}
	\label{fig:unified}
	\end{figure}
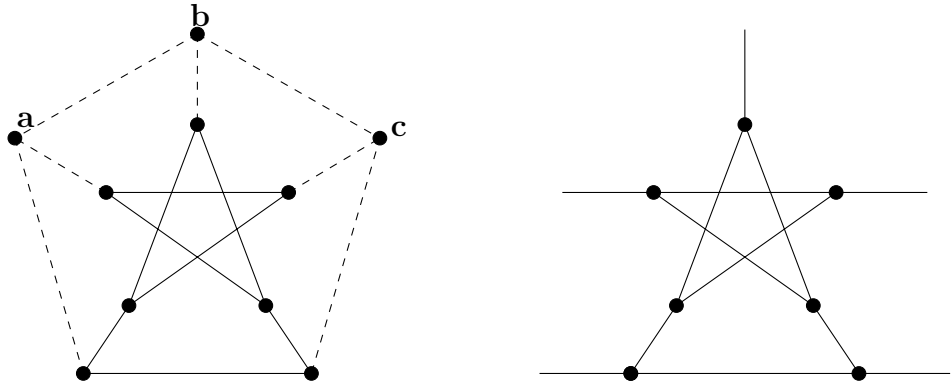

\section{Main result}
\label{sec:mainresult}

Now we construct a normal $5$-edge-colorings for a class of $LP_1$-snarks. Later, we will prove that our class properly contains all snarks considered in \cite{ManSha}.

For any odd $k$, let $G_k^L$ be an $LP_1$-snark, and, as described in Section \ref{sec:aux}, we consider its decomposition in the two subgraphs $G_k^B$ and $G_k^C$. 
Recall that, since $G_k^L$ is an $LP_1$-snark, $G_k^B$ is the disjoint union of copies of $B(P)$ and every connected component of $G_k^C$ is either $K_2$ or a star with four vertices.

Firstly, we construct a normal $5$-edge-coloring of $G_k^B$. We consider the $k$ consecutive copies of $B(P)$, say $B_1,\ldots,B_k$, and we subdivide them in three subsequences $B_1,\dots,B_r$, $B_{r+1},\ldots,B_{r+s}$ and $B_{r+s+1},\ldots,B_{k}$ with $r$ and $s$ odd. Note that all the three subsequences have an odd number of copies of $B(P)$ since $k$ is also odd. 

Color the subgraph of $G_k^B$ induced by the subsequence $B_1,\dots,B_r$ as follows:

	\begin{figure}[H]
		\begin{tikzpicture}[scale=0.54]
		\coordinate (j) at (0,0);
		\coordinate (k)  at (5,0) ;
		\coordinate (v) at (1,1.5) ;
		\coordinate (w) at (4,1.5) ;
		\coordinate (z) at (4.5,4) ;
		\coordinate (x) at (0.5,4) ;
		\coordinate  (y) at (2.5,5.5);
		\draw[fill=black] (x) circle [radius=0.15cm];
		\draw[fill=black] (y) circle [radius=0.15cm];
		\draw[fill=black] (z) circle [radius=0.15cm];
		\draw[fill=black] (w) circle [radius=0.15cm];
		\draw[fill=black] (v) circle [radius=0.15cm];
		\draw[fill=black] (j) circle [radius=0.15cm];
		\draw[fill=black] (k) circle [radius=0.15cm];
		\coordinate (j') at (9,0);
		\coordinate (k') at (14,0);
		\coordinate (v') at (10,1.5);
		\coordinate (w') at (13,1.5);
		\coordinate (z') at (13.5,4);
		\coordinate (x') at (9.5,4);
		\coordinate (y') at (11.5,5.5);
		\draw[fill=black] (j') circle [radius=0.15cm];
		\draw[fill=black] (x') circle [radius=0.15cm];
		\draw[fill=black] (y') circle [radius=0.15cm];
		\draw[fill=black] (z') circle [radius=0.15cm];
		\draw[fill=black] (w') circle [radius=0.15cm];
		\draw[fill=black] (v') circle [radius=0.15cm];
		\draw[fill=black] (j') circle [radius=0.15cm];
		\draw[fill=black] (k') circle [radius=0.15cm];
		\coordinate (j'') at (18+3,0);
		\coordinate (k'') at (23+3,0);
		\coordinate (v'') at (19+3,1.5);
		\coordinate (w'') at (22+3,1.5);
		\coordinate (z'') at (22.5+3,4);
		\coordinate (x'') at (18.5+3,4);
		\coordinate (y'') at (20.5+3,5.5);
		\draw[fill=black] (j'') circle [radius=0.15cm];
		\draw[fill=black] (x'') circle [radius=0.15cm];
		\draw[fill=black] (y'') circle [radius=0.15cm];
		\draw[fill=black] (z'') circle [radius=0.15cm];
		\draw[fill=black] (w'') circle [radius=0.15cm];
		\draw[fill=black] (v'') circle [radius=0.15cm];
		\draw[fill=black] (j'') circle [radius=0.15cm];
		\draw[fill=black] (k'') circle [radius=0.15cm];
		\draw (v)--(j)--(k)--(w)--(y)--(v)--(z)--(x)--(w);
		\draw (v')--(j')--(k')--(w')--(y')--(v')--(z')--(x')--(w');
		\draw (v'')--(j'')--(k'')--(w'')--(y'')--(v'')--(z'')--(x'')--(w'');
		\draw (k)--(j');
		\draw (z)--(x');
		\draw (z')--(16,4);
		\draw (k')--(16,0);
		\draw (19.5,4)--(x'');
		\draw (19,0)--(j'');
		
		\draw (-2,0)--(j);
		\draw (-1.5,4)--(x);
		\draw (y)--(2.5,8);
		\draw (y')--(11.5,8);
		\draw (y'')--(20.5+3,8);
		\draw (k'')--(25+3,0);
		\draw (z'')--(24.5+3,4);
		\node at (4.8,0.8) {$5$};
		\node at (0.2,0.8) {$4$};
		\node at (2.5,0.3) {$3$};
		\node at (-1,0.3) {$1$};
		\node at (7,0.3) {$2$};
		\node at (7,4.3) {$1$};
		\node at (2.1,3.1) {$4$};
		\node at (2.8,3.1) {$5$};
		\node at (2.5,4.3) {$3$};
		\node at (1.6,3.6) {$2$};
		\node at (3.4,3.6) {$1$};
		\node at (-0.5,4.3) {$2$};
		\node at (11.5,4.3) {$3$};
		\node at (16-1.5,4.3) {$2$};
		
		\node at (10.6,3.6) {$1$};
		\node at (12.4,3.6) {$2$};
		\node at (11.1,3.1) {$5$};
		\node at (11.8,3.1) {$4$};
		\node at (13.8,0.8) {$4$};
		\node at (9.2,0.8) {$5$};
		\node at (11.5,0.3) {$3$};
		\node at (16-1,0.3) {$1$};
		
		\node at (20.5+3,0.3) {$3$};
		\node at (18.2+3,0.8) {$4$};
		\node at (22.8+3,0.8) {$5$};
		\node at (24+3,0.3) {$2$};
		\node at (20.8+3,3.1) {$5$};
		\node at (20.1+3,3.1) {$4$};
		\node at (21.4+3,3.6) {$1$};
		\node at (19.6+3,3.6) {$2$};
		\node at (20.5+3,4.3) {$3$};
		\node at (23.5+3,4.3) {$1$};
		\node at (11.2,6.5) {$3$};
		\node at (2.2,6.5) {$3$};
		\node at (20.2+3,6.5) {$3$};
		
		\node at (21-1,0.3)  {$1$};
		\node at (20.7,4.3) {$2$};
		\coordinate (d) at (16,2);
		\draw[dashed]  (d)--(20,2);
		\end{tikzpicture}
		
		\caption{A normal $5$ edge-coloring in the first subsequence of $G_k^B$} \label{Loupekhine first subseq}
	\end{figure}
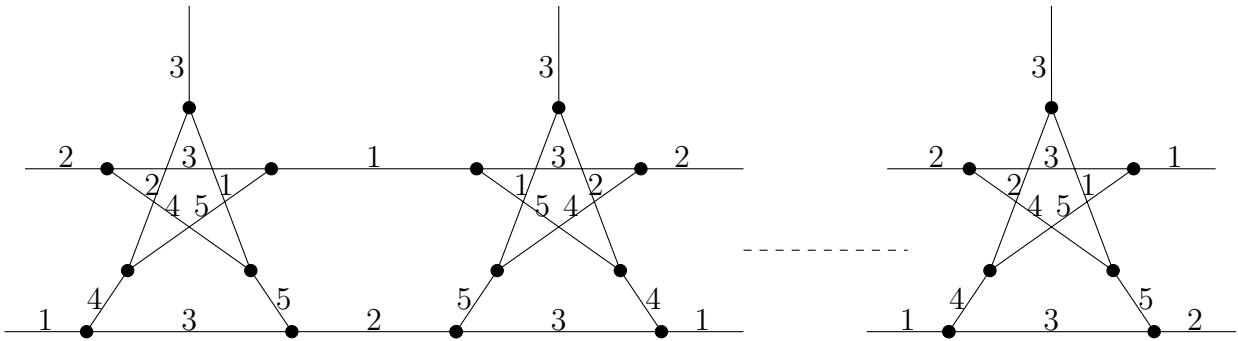	

Color the subgraph of $G_k^B$ induced by the subsequence $B_{r+1},\ldots,B_{r+s}$ as follows:

	\begin{figure}[H]
		\begin{tikzpicture}[scale=0.54]
		\coordinate (j) at (0,0);
		\coordinate (k)  at (5,0) ;
		\coordinate (v) at (1,1.5) ;
		\coordinate (w) at (4,1.5) ;
		\coordinate (z) at (4.5,4) ;
		\coordinate (x) at (0.5,4) ;
		\coordinate  (y) at (2.5,5.5);
		\draw[fill=black] (x) circle [radius=0.15cm];
		\draw[fill=black] (y) circle [radius=0.15cm];
		\draw[fill=black] (z) circle [radius=0.15cm];
		\draw[fill=black] (w) circle [radius=0.15cm];
		\draw[fill=black] (v) circle [radius=0.15cm];
		\draw[fill=black] (j) circle [radius=0.15cm];
		\draw[fill=black] (k) circle [radius=0.15cm];
		\coordinate (j') at (9,0);
		\coordinate (k') at (14,0);
		\coordinate (v') at (10,1.5);
		\coordinate (w') at (13,1.5);
		\coordinate (z') at (13.5,4);
		\coordinate (x') at (9.5,4);
		\coordinate (y') at (11.5,5.5);
		\draw[fill=black] (j') circle [radius=0.15cm];
		\draw[fill=black] (x') circle [radius=0.15cm];
		\draw[fill=black] (y') circle [radius=0.15cm];
		\draw[fill=black] (z') circle [radius=0.15cm];
		\draw[fill=black] (w') circle [radius=0.15cm];
		\draw[fill=black] (v') circle [radius=0.15cm];
		\draw[fill=black] (j') circle [radius=0.15cm];
		\draw[fill=black] (k') circle [radius=0.15cm];
		\coordinate (j'') at (18+3,0);
		\coordinate (k'') at (23+3,0);
		\coordinate (v'') at (19+3,1.5);
		\coordinate (w'') at (22+3,1.5);
		\coordinate (z'') at (22.5+3,4);
		\coordinate (x'') at (18.5+3,4);
		\coordinate (y'') at (20.5+3,5.5);
		\draw[fill=black] (j'') circle [radius=0.15cm];
		\draw[fill=black] (x'') circle [radius=0.15cm];
		\draw[fill=black] (y'') circle [radius=0.15cm];
		\draw[fill=black] (z'') circle [radius=0.15cm];
		\draw[fill=black] (w'') circle [radius=0.15cm];
		\draw[fill=black] (v'') circle [radius=0.15cm];
		\draw[fill=black] (j'') circle [radius=0.15cm];
		\draw[fill=black] (k'') circle [radius=0.15cm];
		\draw (v)--(j)--(k)--(w)--(y)--(v)--(z)--(x)--(w);
		\draw (v')--(j')--(k')--(w')--(y')--(v')--(z')--(x')--(w');
		\draw (v'')--(j'')--(k'')--(w'')--(y'')--(v'')--(z'')--(x'')--(w'');
		\draw (k)--(j');
		\draw (z)--(x');
		\draw (z')--(16,4);
		\draw (k')--(16,0);
		\draw (19.5,4)--(x'');
		\draw (19,0)--(j'');
		
		\draw (-2,0)--(j);
		\draw (-1.5,4)--(x);
		\draw (y)--(2.5,8);
		\draw (y')--(11.5,8);
		\draw (y'')--(20.5+3,8);
		\draw (k'')--(25+3,0);
		\draw (z'')--(24.5+3,4);
		\node at (4.8,0.8) {$2$};
		\node at (0.2,0.8) {$3$};
		\node at (2.5,0.3) {$5$};
		\node at (-1,0.3) {$2$};
		\node at (7,0.3) {$3$};
		\node at (2.1,3.1) {$4$};
		\node at (2.8,3.1) {$5$};
		
		\node at (2.5,4.3) {$2$};
		\node at (7,4.3) {$3$};
		\node at (1.6,3.6) {$2$};
		\node at (3.4,3.6) {$1$};
		\node at (-0.5,4.3) {$1$};
		\node at (11.5,4.3) {$2$};
		\node at (16-1.5,4.3) {$1$};
		
		\node at (10.6,3.6) {$1$};
		\node at (12.4,3.6) {$2$};
		\node at (11.1,3.1) {$5$};
		\node at (11.8,3.1) {$4$};
		\node at (13.8,0.8) {$3$};
		\node at (9.2,0.8) {$2$};
		\node at (11.5,0.3) {$5$};
		\node at (16-1,0.3) {$2$};
		
		\node at (20.5+3,0.3) {$5$};
		\node at (18.2+3,0.8) {$3$};
		\node at (22.8+3,0.8) {$2$};
		\node at (24+3,0.3) {$3$};
		\node at (20.8+3,3.1) {$5$};
		\node at (20.1+3,3.1) {$4$};
		\node at (21.4+3,3.6) {$1$};
		\node at (19.6+3,3.6) {$2$};
		\node at (20.5+3,4.3) {$2$};
		\node at (23.5+3,4.3) {$3$};
		\node at (11.2,6.5) {$4$};
		\node at (2.2,6.5) {$4$};
		\node at (20.2+3,6.5) {$4$};
		
		\node at (21-1,0.3)  {$2$};
		\node at (20.7,4.3) {$1$};
		\coordinate (d) at (16,2);
		\draw[dashed]  (d)--(20,2);
		\end{tikzpicture}
		
		\caption{A normal $5$ edge-coloring in the second subsequence of $G_k^B$} \label{Loupekhine second subseq}
	\end{figure}
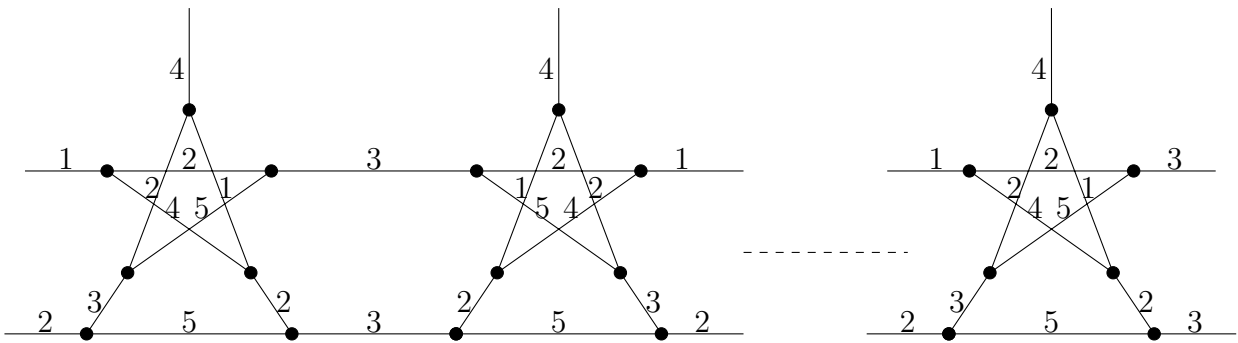	

Color the subgraph of $G_k^B$ induced by the subsequence $B_{r+s+1},\ldots,B_{k}$ as follows:

	\begin{figure}[H]
		\begin{tikzpicture}[scale=0.54]
		\coordinate (j) at (0,0);
		\coordinate (k)  at (5,0) ;
		\coordinate (v) at (1,1.5) ;
		\coordinate (w) at (4,1.5) ;
		\coordinate (z) at (4.5,4) ;
		\coordinate (x) at (0.5,4) ;
		\coordinate  (y) at (2.5,5.5);
		\draw[fill=black] (x) circle [radius=0.15cm];
		\draw[fill=black] (y) circle [radius=0.15cm];
		\draw[fill=black] (z) circle [radius=0.15cm];
		\draw[fill=black] (w) circle [radius=0.15cm];
		\draw[fill=black] (v) circle [radius=0.15cm];
		\draw[fill=black] (j) circle [radius=0.15cm];
		\draw[fill=black] (k) circle [radius=0.15cm];
		\coordinate (j') at (9,0);
		\coordinate (k') at (14,0);
		\coordinate (v') at (10,1.5);
		\coordinate (w') at (13,1.5);
		\coordinate (z') at (13.5,4);
		\coordinate (x') at (9.5,4);
		\coordinate (y') at (11.5,5.5);
		\draw[fill=black] (j') circle [radius=0.15cm];
		\draw[fill=black] (x') circle [radius=0.15cm];
		\draw[fill=black] (y') circle [radius=0.15cm];
		\draw[fill=black] (z') circle [radius=0.15cm];
		\draw[fill=black] (w') circle [radius=0.15cm];
		\draw[fill=black] (v') circle [radius=0.15cm];
		\draw[fill=black] (j') circle [radius=0.15cm];
		\draw[fill=black] (k') circle [radius=0.15cm];
		\coordinate (j'') at (18+3,0);
		\coordinate (k'') at (23+3,0);
		\coordinate (v'') at (19+3,1.5);
		\coordinate (w'') at (22+3,1.5);
		\coordinate (z'') at (22.5+3,4);
		\coordinate (x'') at (18.5+3,4);
		\coordinate (y'') at (20.5+3,5.5);
		\draw[fill=black] (j'') circle [radius=0.15cm];
		\draw[fill=black] (x'') circle [radius=0.15cm];
		\draw[fill=black] (y'') circle [radius=0.15cm];
		\draw[fill=black] (z'') circle [radius=0.15cm];
		\draw[fill=black] (w'') circle [radius=0.15cm];
		\draw[fill=black] (v'') circle [radius=0.15cm];
		\draw[fill=black] (j'') circle [radius=0.15cm];
		\draw[fill=black] (k'') circle [radius=0.15cm];
		\draw (v)--(j)--(k)--(w)--(y)--(v)--(z)--(x)--(w);
		\draw (v')--(j')--(k')--(w')--(y')--(v')--(z')--(x')--(w');
		\draw (v'')--(j'')--(k'')--(w'')--(y'')--(v'')--(z'')--(x'')--(w'');
		\draw (k)--(j');
		\draw (z)--(x');
		\draw (z')--(16,4);
		\draw (k')--(16,0);
		\draw (19.5,4)--(x'');
		\draw (19,0)--(j'');
		
		\draw (-2,0)--(j);
		\draw (-1.5,4)--(x);
		\draw (y)--(2.5,8);
		\draw (y')--(11.5,8);
		\draw (y'')--(20.5+3,8);
		\draw (k'')--(25+3,0);
		\draw (z'')--(24.5+3,4);
		\node at (4.8,0.8) {$3$};
		\node at (0.2,0.8) {$1$};
		\node at (2.5,0.3) {$4$};
		\node at (-1,0.3) {$3$};
		\node at (7,0.3) {$1$};
		\node at (2.1,3.1) {$4$};
		\node at (2.8,3.1) {$5$};
		
		\node at (2.5,4.3) {$1$};
		\node at (7,4.3) {$2$};
		\node at (1.6,3.6) {$2$};
		\node at (3.4,3.6) {$1$};
		\node at (-0.5,4.3) {$3$};
		\node at (11.5,4.3) {$1$};
		\node at (16-1.5,4.3) {$2$};
		
		\node at (10.6,3.6) {$1$};
		\node at (12.4,3.6) {$2$};
		\node at (11.1,3.1) {$5$};
		\node at (11.8,3.1) {$5$};
		\node at (13.8,0.8) {$1$};
		\node at (9.2,0.8) {$2$};
		\node at (11.5,0.3) {$5$};
		\node at (16-1,0.3) {$2$};
		
		\node at (20.5+3,0.3) {$5$};
		\node at (18.2+3,0.8) {$1$};
		\node at (22.8+3,0.8) {$2$};
		\node at (24+3,0.3) {$1$};
		\node at (20.8+3,3.1) {$5$};
		\node at (20.1+3,3.1) {$5$};
		\node at (21.4+3,3.6) {$1$};
		\node at (19.6+3,3.6) {$2$};
		\node at (20.5+3,4.3) {$1$};
		\node at (23.5+3,4.3) {$2$};
		\node at (11.2,6.5) {$5$};
		\node at (2.2,6.5) {$5$};
		\node at (20.2+3,6.5) {$5$};
		
		\node at (21-1,0.3)  {$2$};
		\node at (20.7,4.3) {$2$};
		\coordinate (d) at (16,2);
		\draw[dashed]  (d)--(20,2);
		\end{tikzpicture}
		
		\caption{A normal $5$ edge-coloring in the third subsequence of $G_k^B$} \label{Loupekhine third subseq}
	\end{figure}
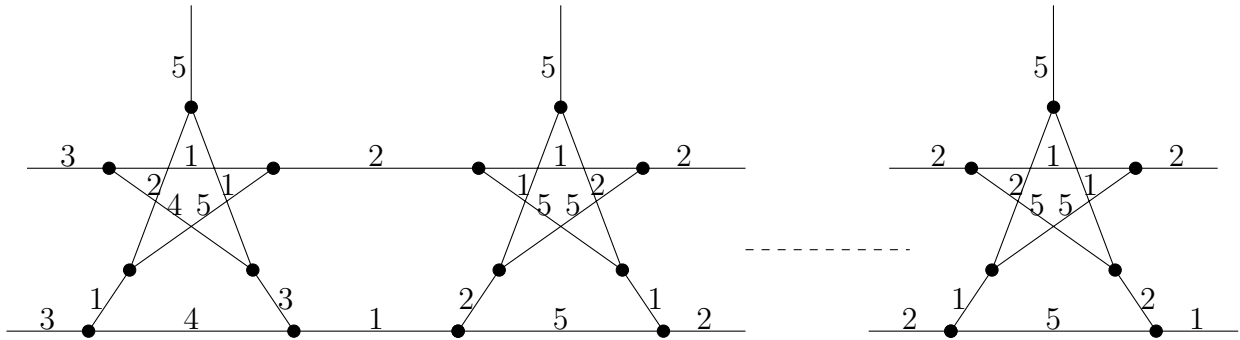	

By a direct check, the defined coloring is normal for each of the previous subgraphs and it remains normal also when we consider the entire subgraph $G_k^B$. Moreover, all the $k$ vertices of degree $2$ are incident to edges with colors $1$ and $2$. 

Now, consider a subgraph $G_k^C$ which satisfies the additional property that {\bf every copy of $K_2$ connects two vertices of $G_k^B$ in the same subsequence.}

Color every edge of $G^C_k$ incident a vertex of the first subsequence with the color $3$, incident a vertex of the second subsequence with the color $4$ and incident a vertex of the third subsequence with the color $5$. 

For all $G_k^C$ which satisfy the additional constraint, the presented coloring is a normal $5$-edge-coloring of $G_k^L$. 
Then, we are now in position to state our main result.

\begin{theorem}\label{thm:mainresult}
Let $G_k^L$ be an $LP_1$-snark and let $B_1 \dots B_k$ be the $k$ blocks in $G_k^B$.
If the $k$ blocks can be partitioned in three subsequences of odd length in such a way that every copy of $K_2$ in $G_k^C$ connects two vertices of $G_k^B$ in the same subsequence, then $G_k^L$ admits a normal $5$-edge-coloring.
\end{theorem}

Now, let us show that the class of $LP_1$-snarks presented in \cite{ManSha} satisfy the assumptions of Theorem \ref{thm:mainresult}. 

For any odd $k$, the $LP_1$-snark considered in \cite{ManSha} is the one having the following decomposition in $G_k^C$ and $G_k^B$.

The subgraph $G_k^C$ has a unique star and denote its three vertices of degree one by $z_{1}, z_{\lfloor{k/2}\rfloor+1}$ and $z_{\lfloor{k/2}\rfloor+2}$. Moreover, each copy of $K_2$ connects pairwise the vertices $z_{i+2},z_{k-i}$, for $i=0, \dots \lfloor{k/2}\rfloor-2$. 

We use the same patterns of coloring in Figures \ref{Loupekhine first subseq},
\ref{Loupekhine second subseq}, \ref{Loupekhine third subseq} where  $B_{\lfloor{k/2}\rfloor+3},...,B_{1},...B_{\lfloor{k/2}\rfloor}$ is the first subsequence, the unique block $B_{\lfloor{k/2}\rfloor+1}$ gives the second subsequence and the unique block $B_{\lfloor{k/2}\rfloor+2}$ is the third one.

Hence, every element of this family of $LP_1$-snarks satisfies the assumptions of Theorem \ref{thm:mainresult} and then it is contained in our larger class of $LP_1$-snarks. 

\subsection{Loupekhine snarks of the second type}

Here we briefly explain how we can slightly modify the coloring presented above in order to exhibit a normal $5-$edge-coloring also for the class of Loupekhine snarks of the second kind \cite{construction Loupekhine}. 
First of all, let us define this class. We will call \textit{twist} the operation that replace two edges between two consecutive blocks in $G^B_k$ in the way described in Figure \ref{Loupekhine second}.

\begin{definition}
An $LP_2$-snark is a snark which is obtained by an $LP_1$-snark with a twist of the two edges between two consecutive selected blocks.
\end{definition}

\begin{flushright}
	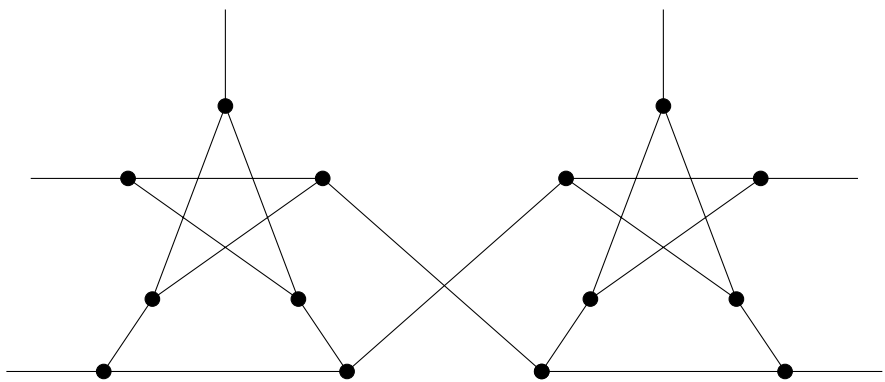
\begin{figure}[!htbp]
	\begin{center}
		\begin{tikzpicture}[scale=0.64]
		\coordinate (j) at (0,0);
		\coordinate (k)  at (5,0) ;
		\coordinate (v) at (1,1.5) ;
		\coordinate (w) at (4,1.5) ;
		\coordinate (z) at (4.5,4) ;
		\coordinate (x) at (0.5,4) ;
		\coordinate  (y) at (2.5,5.5);
		\draw[fill=black] (x) circle [radius=0.15cm];
		\draw[fill=black] (y) circle [radius=0.15cm];
		\draw[fill=black] (z) circle [radius=0.15cm];
		\draw[fill=black] (w) circle [radius=0.15cm];
		\draw[fill=black] (v) circle [radius=0.15cm];
		\draw[fill=black] (j) circle [radius=0.15cm];
		\draw[fill=black] (k) circle [radius=0.15cm];
		\coordinate (j') at (9,0);
		\coordinate (k') at (14,0);
		\coordinate (v') at (10,1.5);
		\coordinate (w') at (13,1.5);
		\coordinate (z') at (13.5,4);
		\coordinate (x') at (9.5,4);
		\coordinate (y') at (11.5,5.5);
		\draw[fill=black] (j') circle [radius=0.15cm];
		\draw[fill=black] (x') circle [radius=0.15cm];
		\draw[fill=black] (y') circle [radius=0.15cm];
		\draw[fill=black] (z') circle [radius=0.15cm];
		\draw[fill=black] (w') circle [radius=0.15cm];
		\draw[fill=black] (v') circle [radius=0.15cm];
		\draw[fill=black] (j') circle [radius=0.15cm];
		\draw[fill=black] (k') circle [radius=0.15cm];
		\draw (v)--(j)--(k)--(w)--(y)--(v)--(z)--(x)--(w);
		\draw (v')--(j')--(k')--(w')--(y')--(v')--(z')--(x')--(w');
		\draw (k)--(x');
		\draw (z)--(j');
		\draw (z')--(15.5,4);
		\draw (k')--(16,0);
		\draw (-2,0)--(j);
		\draw (-1.5,4)--(x);
		\draw (y)--(2.5,7.5);
		\draw (y')--(11.5,7.5);
		
		\end{tikzpicture}
		\end{center}
		
		\caption{A twist between two copies of $B(P)$.} 
		\label{Loupekhine second}
	\end{figure}
		
\end{flushright}	
Observe that we do not really construct new examples of snarks if we perform additional twists between other blocks of the snark so obtained. Indeed, it is easy to prove that an even number of twists produces a graph isomorphic to the original one. Hence, by applying an odd number of twists we obtain the same snark obtained by using a unique twist. 
Moreover, it is not relevant the pair of blocks where the twist is applied: indeed, given an $LP_1$-snark $G^L_k$, all $LP_2$-snarks obtained starting from $G_k^L$ with a twist between two arbitrary consecutive blocks are isomorphic.

Hence, without loss of generality, we can always apply the twist to the two edges with color $3$ between the last
block in the second subsequence and the first block in the third
subsequence. With such a configuration, we maintain clearly the same pattern of coloring described before for the entire graph and the following holds.

\begin{theorem}
Let $G_k^L$ be an $LP_2$-snark and let $B_1 \dots B_k$ be the $k$ blocks in $G_k^B$.
If the $k$ blocks can be partitioned in three subsequences of odd length in such a way that every copy of $K_2$ in $G_k^C$ connects two vertices of $G_k^B$ in the same subsequence, then $G_k^L$ admits a normal $5$-edge-coloring.
\end{theorem}

We would like to stress that it remains largely open the general problem of proving that the normal chromatic number of any Loupekhine snark is five, and it is still open even if we restrict our attention to the class of $LP_1$-snarks.



\bibliographystyle{elsarticle-num}


\end{document}